\documentclass{extarticle}
\usepackage{amsfonts,amssymb,amsmath,amsthm}
\usepackage{url}
\usepackage{color}
\usepackage{nameref}
\usepackage[dvipsnames]{xcolor}
\usepackage[colorlinks=true,linkcolor=RoyalBlue,citecolor=PineGreen,urlcolor=blue]{hyperref}
\usepackage[margin=3cm,letterpaper]{geometry}
\usepackage{bm}
\usepackage[ruled,vlined,linesnumbered]{algorithm2e}
\usepackage[numbers,sort]{natbib}
        \RequirePackage{enumerate} 
          \usepackage[shortlabels]{enumitem}

\usepackage{graphicx}
\usepackage{float}
\usepackage{subfig}
\usepackage{cleveref}

\crefformat{footnote}{#2\footnotemark[#1]#3}

\newtheorem{lemma}{Lemma}[section]

\newtheorem{theorem}{Theorem}[section]
\newtheorem{corollary}{Corollary}[section]

\newtheorem{thm}{Theorem}[section]
\newtheorem{proposition}{Proposition}[section]

\theoremstyle{definition}
\newtheorem{example}{Example}[section]
\newtheorem{remark}{Remark}[section]
\newtheorem{definition}{Definition}[section]
\newtheorem{notation}{Notation}[section]

\DeclareMathOperator{\AS}{AS}

\makeatletter
\newcommand{\myitem}[1]{%
\item[#1]\protected@edef\@currentlabel{#1}%
}
\makeatother

\makeatletter
\newcommand\footnoteref[1]{\protected@xdef\@thefnmark{\ref{#1}}\@footnotemark}
\makeatother

\title{The Dynamics of Canalizing Boolean Networks}

\author{Elijah Paul\footnote{{eligpaul@gmail.com}, California Institute of Technology, Pasadena, USA} \textsuperscript{, $\S$}, Gleb Pogudin\footnote{{pogudin.gleb@gmail.com}, Department of Computer Science, National Research University Higher School of Economics, Moscow, Russia and Courant Institute of Mathematical Sciences, New York University, New York,  USA} \textsuperscript{, $\S$}, William Qin\footnote{{wqin2008@gmail.com}, Millburn High School, Millburn, USA} \textsuperscript{,}\footnote{The first three authors are ordered alphabetically}, Reinhard Laubenbacher\footnote{{laubenbacher@uchc.edu}, Center for Quantitative Medicine, University of Connecticut Health Center and Jackson Laboratory for Genomic Medicine, Farmington, USA}}

\date{}

\begin{document}

\maketitle

\begin{abstract}
    Boolean networks are a popular modeling framework in computational biology to capture the dynamics of molecular networks, such as gene regulatory networks. 
    It has been observed that many published models of such networks are defined by regulatory rules driving the dynamics that have certain so-called canalizing properties. 
    In this paper, we investigate the dynamics of a random Boolean network with such properties using analytical methods and simulations.
    
    From our simulations, we observe that
    Boolean networks with higher canalizing depth have generally fewer attractors, the attractors are smaller, and the basins are larger, with implications for the stability and robustness of the models.
    These properties are relevant to many biological applications.
    Moreover, our results show that, from the standpoint of the attractor structure, high canalizing depth, compared to relatively small positive canalizing depth, has a very modest impact on dynamics.
    
    Motivated by these observations, we conduct mathematical study of the attractor structure of a random Boolean network of canalizing depth one (i.e., the smallest positive depth).
    For every positive integer $\ell$, we give an explicit formula for the limit of the expected number of attractors of length $\ell$ in an $n$-state random Boolean network as $n$ goes to infinity.
\end{abstract}

\section{Introduction}

Dynamic mathematical models are a key enabling technology in systems biology. Depending on the system to be 
modeled, the data and information available for their construction, and the questions to be answered, different 
modeling frameworks can be used. For kinetic models, systems of ordinary differential equations have a long
tradition. Generally, they will have the very special structure of polynomial equations representing Michaelis-Menten kinetics,
even in the case of systems, such as gene regulatory networks, that are not properly biochemical reaction
networks. It is this special structure that gives models desirable properties and aids in model analysis. 
Besides continuous models, a range of discrete models are finding increasingly frequent use, in particular 
Boolean network models of a broad variety of biological systems, from intracellular molecular networks to 
population-level compartmental models (see, e.g., \cite{DB08, KPST03,ML09,VS11,SR07}), 
going back to the work of S. Kauffman in the 1960s~\cite{K67, K69, K69b}. 
While Boolean network models, a collection of nodes, whose regulation by other nodes is described via a 
logical rule built from Boolean operators, are intuitive and mathematically simple to describe, their analysis is 
severely limited by the lack of mathematical tools. It generally consists of simulation results. 
Any set function on binary strings that takes on binary values
can be represented as a Boolean function, so that the class of general Boolean networks is identical to the 
class of set functions on binary strings of a given length, making any general analysis impossible. The search 
for special classes of Boolean functions that are broad enough to cover all or most rules that occur in biology,
but special enough to allow for mathematical approaches has a long history. 

It was again S. Kauffman who proposed a class of functions~\cite{K69} with properties inspired by the developmental biology
concept of \emph{canalization}, going back to C. Waddington in the 1940s~\cite{W42}. There is some evidence that 
canalizing Boolean functions do indeed appear disproportionately in published models, and that the dynamics
of Boolean network models consisting of canalizing functions has special properties, in particular
a "small" number of attractors. This is important since, in the case of intracellular molecular network models, 
attractors correspond to the different phenotypes a cell is capable of. Here, again, the majority of 
available results are obtained by simulating large numbers of such networks. 
The main question of this paper is:\emph{What do the dynamics of a random canalizing Boolean network look like?}
We approach this question using both computer simulations and analytical methods, with the main result of the 
paper being Theorem~\ref{thm:main}, which gives a \textbf{provable} formula for the number of expected attractors
of a general Boolean network with a particular canalization property. In addition to providing important information
about canalizing Boolean network models, this result can be viewed as a part of a growing body of mathematical
results characterizing this class of networks that promises to be as rich as that for chemical reaction network models
based on ordinary differential equations.

\section{Background}

The property of canalization for Boolean functions was introduced by S. Kauffman in~\cite{K69}, inspired by the concept of \textit{canalization} from developmental biology~\cite{W42}.
A Boolean function is canalizing if there is a variable and a value of the variable such that if the variable takes the value, then the value of the function does not depend on other variables.
It was shown that models defined by such functions often exhibit less chaotic and more stable behavior~\cite{KPST04, KH07}.
Nested canalizing functions, obtained by applying the concept of canalization recursively, were introduced in~\cite{KPST03}.
They form a special subset of canalizing functions and have stable dynamics~\cite{KPST04}.
We note that there are other important properties shared by Boolean networks arising in modeling (for example, sparsity~\cite{K69}). 
In this paper we focus only on canalization and its impact on the dynamics, and one of the natural future directions would be to consider several such properties simultaneously.

To cover more models arising in applications, the notion of nested canalizing function was relaxed by Layne, Dimitrova, and Macaulay~\cite{LDM12} by assigning to every Boolean function its \emph{canalizing depth}.
Non-canalizing functions have canalizing depth zero and nested canalizing functions have the maximal possible canalizing depth, equal to the number of variables.
Canalizing depth of a Boolean network is defined as the minimum of the canalizing depths of the functions defining the network.
In~\cite{LDM12}, activities and sensitivities of functions with different canalizing depths and stability and criticality of Boolean networks composed from such functions were investigated.
It has been observed that Boolean networks of higher canalizing depth tend to be more stable and less sensitive. 
However, increasing the canalizing depth to the maximum does not improve the stability significantly compared to moderate positive canalizing depth.  
These observations give a strong indication of the biological utility of canalizing function, even with small canalizing depth.

Attractors in Boolean network models can be interpreted as distinct cell types~\cite[p.~202]{K93} and their lengths can be viewed as the variety of different gene expression patterns corresponding to the cell type.
Thus, understanding the attractor structure of a random Boolean network defined by functions of a fixed canalizing depth is important for assessing biological relevance of such models. 
Analytic study of the attractor structure of nested canalizing Boolean networks has been done in~\cite{KPST04}.
For discussion about attractors of length one (i.e., steady state), we refer to~\cite{VAHL14}.

\section{Our results.}
The main question of this paper is:\emph{What do the dynamics of a random canalizing Boolean network look like?}
We approach this question using both computer simulations and analytical methods.

In our \textbf{computational experiments}, we generate approximately $30$ million random Boolean networks of all possible canalizing depths with the number of variables ranging from $4$ to $20$.
For each of these networks, we determine sizes of all the attractors and basins of attraction and analyze the obtained data.
We discover the following:
\begin{enumerate}[label=\textbf{(\arabic*)}]
    \item\label{obs:decrease} For a fixed number of variables, the sample mean of the number of attractors and average size of an attractor decrease when the canalizing depth increases. 
    \item The decrease of the average size of an attractor is much greater than the decrease of the number of attractors as the canalizing depth increases.
    \item\label{obs:the_same} Both decreases from~\ref{obs:decrease} are substantial when the canalizing depth changes from zero to small canalizing depths but a further increase of the canalizing depth does not lead to a significant decrease for either the sample means or for the empirical distributions.
    \item\label{obs:larger_n} The relative decrease of the sample mean of the number of attractors and the average attractor size when the canalizing depth changes from zero to one becomes sharper when the number of variables increases.
\end{enumerate}

The observations~\ref{obs:decrease} and~\ref{obs:the_same} are consistent with the results obtained in~\cite{LDM12} for sensitivity and stability.
This provides new evidence that Boolean networks of small positive canalizing depth are almost as well-suited for modeling as those with nested canalizing functions, from the point of view of stability.
Since there are many more canalizing functions of small positive canalizing depth than nested canalizing functions~\cite[Section~5]{HM16}, they provide a richer modeling toolbox.

Motivated by the observation~\ref{obs:the_same}, we conduct a \textbf{mathematical study} of the attractor structure of a random Boolean network of canalizing depth one (that is, the minimal positive depth).
Our main theoretical result, Theorem~\ref{thm:main}, gives, for every positive integer $\ell$, a formula for the limit of the expected number of attractors of length $\ell$ in a random Boolean network of depth one. 
The same formulas are valid for a random Boolean network defined by canalizing functions (see Remark~\ref{rem:at_least_one}).
In particular, our formulas show that a large random network of depth one, on average, has more attractors of small sizes that an average Boolean network (Remark~\ref{rem:compare_with_depth_zero}).

Formulas similar to the ones in our proofs (e.g., in Lemma~\ref{lem:represent_B}) have already appeared in the study of the average number of attractors of a given length in sparse Boolean networks, e.g.~\cite[Eq. (2)]{ST2003} and~\cite[Eq. (6)]{Drossel}.
The results of~\cite{ST2003} and~\cite{Drossel} are based on describing the asymptotic behavior of these formulas in terms of $N$, the number of nodes in the network, and the asymptotics are of the form $\mathrm{O}(N^\alpha)$.
In our case, the average number of attractors of a given length simply approaches a constant as $N \to \infty$ (that is, is $\mathrm{O}(1)$), but our methods allow us to find the exact value of this constant.

The source code we used for generating and analyzing data is available at~\url{https://github.com/MathTauAthogen/Canalizing-Depth-Dynamics}.
The raw data is available at~\url{https://github.com/MathTauAthogen/Canalizing-Depth-Dynamics/tree/master/data}.

\paragraph{Structure of the paper.} The rest of the paper is organized as follows. Section~\ref{sec:preliminaries} contains necessary definitions about canalizing functions and Boolean networks.
Outlines of the algorithms used in our computational experiments are in Section~\ref{sec:algorithms}.
The main observations are summarized in Section~\ref{subsec:results}.
Our main theoretical result about attractors in a random Boolean network of canalizing depth one (Theorem~\ref{thm:main}) is presented in Section~\ref{sec:th_result}.
Section~\ref{sec:conclusions} contains conclusions.
The proofs are located in the Appendix.


\section{Preliminaries}\label{sec:preliminaries}

\begin{definition}
A \emph{Boolean network} is a tuple $\mathbf{f} = (f_1, f_2, \ldots ,f_n)$ of Boolean functions in $n$ variables. 
For a state $\bm{a}_t = (a_{t,1}, a_{t,2}, \ldots,  a_{t,n}) \in \{0, 1\}^n$ at time $t$, 
we define the state  $\bm{a}_{t+1} :    = \mathbf{f}(\bm{a}_t) =  (a_{t+1,1},\ldots ,a_{t+1,n}) \in \{0, 1\}^n$ at time $t + 1$ by
\begin{align*}
&a_{t+1,1}=f_1(a_{t, 1},\ldots,a_{t, n}),\\
&\vdots\\
&a_{t+1,n}=f_n(a_{t, 1},\ldots, a_{t, n}).
\end{align*}
\end{definition}

\begin{definition}[Attractors and basins]
Let $\mathbf{f} = (f_1, \ldots, f_n)$ be a Boolean network.
\begin{itemize}
    \item A sequence $\bm{a}_1, \ldots, \bm{a}_\ell \in \{0, 1\}^n$ of distinct states is called an \emph{attractor} of $\mathbf{f}$ if $\mathbf{f}(\bm{a}_i) = \bm{a}_{i + 1}$ for every $1 \leqslant i < \ell$ and $\mathbf{f}(\bm{a}_\ell) = \bm{a}_1$.
    \item An attractor $\bm{a}_1, \ldots, \bm{a}_\ell \in \{0, 1\}^n$ is called \emph{a steady state} if $\ell = 1$.
    \item Let $A = (\bm{a}_1, \ldots, \bm{a}_\ell) \in \left(\{0, 1\}^{n}\right)^\ell$ be an attractor of $\mathbf{f}$.
    The \emph{basin} of $A$ is the set
    \[
    \{ \bm{b} \in \{0, 1\}^n \text{ }| \text{ }\exists N \; : \; \underbrace{\mathbf{f}(\mathbf{f}(\ldots(\mathbf{f}(}_{N \text{ times}}\bm{b})\ldots)) \in A \}.
    \]
\end{itemize}
\end{definition}

\begin{definition}
  A nonconstant function $f(x_1,\ldots x_n)$ is \emph{canalizing} with respect to a variable $x_i$ if there exists a \emph{canalizing value} $a \in \{0, 1\}$ such that 
  \[
    f(x_1, \ldots, x_{i - 1}, a, x_{i + 1}, \ldots, x_n) \equiv \text{const}.
  \] 
\end{definition}
  
  \begin{example}
  Consider $f(x_1,x_2)=x_1 \cdot x_2$ (the product is understood modulo 2, that is, logical AND).
  It is canalizing with respect to $x_1$ with canalizing value $0$, because $f(0, x_2) = 0$ regardless of the value of $x_2$. 
  Analogously, it is canalizing with respect to $x_2$ with canalizing value $0$.
  
  Consider $g(x_1,x_2)=x_1 + x_2$ (summation is understood modulo 2, that is, logical XOR).
  It is not canalizing with respect to $x_1$, because
  \[
    g(0, x_2) = x_2 \neq \operatorname{const} \quad\text{ and }\quad g(1, x_2) = \bar{x}_2 \neq \operatorname{const}.
  \]
  The same argument works for $x_2$ as well.
  \end{example}

\begin{definition}\label{def:can_depth}
$f(x_1,\ldots, x_n)$ has \emph{canalizing depth}~\cite[Definition~2.3]{HM16} $k$ if it can be expressed as 
\[
f=
\begin{cases}
b_1 & x_{i_1}= a_1\\
b_2 & x_{i_1} \not= a_1, x_{i_2}= a_2\\
\vdots\\
b_k & x_{i_1} \not=a_1, x_{i_2}\not= a_2 \ldots x_{i_{k-1}}\not=a_{k-1}, x_{i_k}=a_k\\
g\not\equiv b_k & x_{i_1}\neq a_1,\ldots,x_{i_k}\neq a_k,
\end{cases}
\]
where
\begin{itemize}[itemsep=0pt]
  \item $i_1, \ldots, i_k$ are distinct integers from $1$ to $n$;
  \item $a_1, \ldots, a_k, b_1, \ldots, b_k \in \{0, 1\}$;
  \item $g$ is a noncanalizing function in the variables $\{x_1, \ldots, x_n\} \setminus \{x_{i_1}, \ldots, x_{i_k}\}$. 
\end{itemize}
\end{definition}

\begin{example}
For example, if $f(x_1,x_2,x_3) = (x_1 + x_2)x_3$, 
\[
f(x_1, x_2, x_3) =
\begin{cases}
0, & x_3 = 0\\
x_1 + x_2, & x_3 \not= 0\\
\end{cases}
\] 
and $x_1 + x_2$ is noncanalizing. 
Therefore $f$ has canalizing depth 1.
\end{example}

\begin{remark}\label{rem:multiplecanalizing} 
Since $g$ in Definition~\ref{def:can_depth} is noncanalizing, every function has a single well-defined canalizing depth.
In particular, a function of depth two is not considered to have depth one.
\end{remark}

\begin{definition}\label{def:nested}
We say that a canalizing Boolean function $f(x_1,\ldots,x_n)$ is \emph{nested} if $f$ has canalizing depth $n$; that is, $g = 0$ or $g = 1$ (see Definition~\ref{def:can_depth}). 
For example, $f(x_1,x_2,x_3)= x_1x_2x_3$ is nested canalizing because 
\[
f=
\begin{cases}
0 & x_3=0\\
0 & x_3 \not= 0, x_2= 0\\
0 & x_2,x_3\neq 0, x_1 = 0\\
1 & x_1,x_2,x_3\neq 0
\end{cases}
\] so the canalizing depth of $f$ is 3, which is equal to $n = 3$.
\end{definition}

\begin{definition}\label{def:network_depth}
We say that a Boolean network $\mathbf{f} = (f_1, \ldots, f_n)$ has \emph{canalizing depth} $k$ if $f_1, \ldots, f_n$ are Boolean functions of canalizing depth $k$.
\end{definition}


\section{Simulations: outline of the algorithms}\label{sec:algorithms}

In our computational experiment, we generated random Boolean networks of various canalizing depths.
For each network, we store a list of pairs $(a_i, b_i)$, where $a_i$ is the size of the $i$-th attractor of the network, and $b_i$ is the size of its basin. 
The generated data is available at~\url{https://github.com/MathTauAthogen/Canalizing-Depth-Dynamics/tree/master/data}.
To generate the data, we used two algorithms: one for generating a random Boolean network of a given canalizing depth and one for finding the sizes of attractors and their basins.

\subsection{Finding the sizes of the attractors and their basins}

\begin{algorithm}[H]
\caption{Finding the sizes of the attractors and their basins}\label{alg:sizes}
\vspace{0.1in}
\begin{description}[leftmargin=2.5em]
  \item[In:] A Boolean network $\mathbf{f} = (f_1, \ldots, f_n)$ in $n$ variables 
 \item[Out:]
     A list of pairs $(a_i, b_i)$, where $a_i$ is the size of the $i$-th attractor of $\mathbf{f}$ and $b_i$ is the size of its basin
\end{description}

\vspace{-0.2in}

\begin{enumerate}[leftmargin=!,labelwidth=1.5em,itemsep=0.0in, label=\arabic*]

\item{[Network $\to$ Graph]} Build a directed graph $G$ with $2^n$ vertices corresponding to possible states and a directed edge from $\bm{a}$ to $\mathbf{f}(\bm{a})$ for every $\bm{a} \in \{0, 1\}^n$.

\item{[Attractors]}\label{step:attractors} Perform a depth-first search~\cite[\S~22.3]{CLRS} traversal on $G$ viewed as an undirected graph to detect the unique cycle in each connected component, these cycles are the attractors.

\item{[Basins]}\label{step:basins} For each cycle from Step~\ref{step:attractors}, perform a depth-first search traversal on $G$ with all the edges reversed.
The dfs trees will be the basins.

\item Return the sizes of the attractors and basins found on Steps~\ref{step:attractors} and~\ref{step:basins}.
\end{enumerate}
\end{algorithm}

\subsection{Generating random Boolean functions of a given canalizing depth}

\cite[Section~5]{LDM12} contains a sketch of an algorithm for generating random Boolean functions that have canalizing depth at least $k$ for a given $k$.
Here, we generate functions of canalizing depth equal to $k$ and take a different approach than~\cite{LDM12}.
In order to ensure that the probability distribution of possible outputs is uniform, we use the following structure theorem due to He and Macaulay~\cite{HM16}.

\begin{thm}[{{\cite[Theorem~4.5]{HM16}}}]\label{thm:k_canalizing}
Every Boolean function $f(x_1,\ldots,x_n)\not\equiv 0$ can be uniquely written as
\begin{equation}\label{eqn:layers}
f(x_1,\ldots,x_n)=M_1(M_2(\cdots(M_{r-1}(M_rp_C+1)+1)\cdots)+1)+b,
\end{equation}
where $\displaystyle M_i=\prod_{j = 1}^{k_i}(x_{i_j} + a_{i_j})$ for every $1 \leqslant i \leqslant r$, $p_C\not\equiv 0$ is a noncanalizing function, and $k = \sum\limits_{i = 1}^r k_i$ is the canalizing depth. 
Each $x_i$ appears in exactly one of $\{M_1, \dots, M_r, p_C\}$, and the only restrictions on Eq.~\eqref{eqn:layers} are the following ``exceptional cases'':
\begin{enumerate}[label = (E\arabic*), itemsep=0pt, topsep=1pt]
\item\label{e1} If $p_C\equiv 1$ and $r\neq 1$, then $k_r\geq 2$;
\item\label{e2} If $p_C\equiv 1$ and $r=1$ and $k_1=1$, then $b=0$.
\end{enumerate}
\end{thm}

\begin{example}
Consider $f(x_1,x_2,x_3,x_4) = x_1(x_2 + 1)(x_3x_4 + x_3 + x_4)$ can be represented as 
\[
f = ((x_1 + 0)(x_2 + 1))(((x_3 + 1)(x_4 + 1))(1) + 1) + 0, 
\]
so $M_1=(x_1+0)(x_2+1)$, $M_2=(x_3+1)(x_4+1)$, $b=0$, $k=4$, and $p_C=1$. 
This can be verified by expanding the brackets in the original and new representations of $f$.

Consider $g(x_1,x_2,x_3,x_4,x_5) = 1 + x_5 (x_1 + x_2)(x_3 + 1)x_4$.
It can be represented as 
\[
g = (x_5+0)(((x_3+1)(x_4+0))(x_1 + x_2) + 1) + 1, 
\]
so $M_1=(x_5+0)$, $M_2=(x_4+0)(x_3+1)$, $b=1$, $k=3$, and $p_C=x_1 + x_2$.
\end{example}

Our algorithm is summarized in Algorithms~\ref{alg:k_canalizing} and~\ref{alg:partition} below. 
Correctness of Algorithm~\ref{alg:k_canalizing} follows from Theorem~\ref{thm:k_canalizing}, and correctness of Algorithm~\ref{alg:partition} can be proved directly by induction on $k$.

\begin{remark}
The complexity of Algorithm~\ref{alg:k_canalizing} is $\mathrm{O}(n2^n)$ (see Proposition~\ref{prop:alg_k_can}). 
Given that the size of the output is $\mathrm{O}(2^n)$, this is nearly optimal.

We measured the runtimes of our implementation of Algorithm~\ref{alg:k_canalizing} on a laptop with a Core i5 processor (1.60GHz) and 8Gb RAM.
Generating a single function with $20$ variables (the largest number we used in our simulations) takes $4.9-5.5$ seconds (faster for smaller canalizing depth).
On a laptop, our implementation can go up to $24$ variables ($\sim 2$ minutes to generate a function), and then hits memory limits.
One can go further by using a lower level language and more careful packing.
However, already a Boolean function in $40$ variables would require at least 128Gb of memory.
\end{remark}

\begin{algorithm}[H]
\caption{Generating a random Boolean function of a given canalizing depth}\label{alg:k_canalizing}
\begin{description}[leftmargin=2.5em]
  \item[In:] Nonnegative integers $k$ and $n$ with $k \leqslant n$ 
 \item[Out:] A Boolean function $f$ in $n$ variables of canalizing depth $k$ such that, for fixed $k$ and $n$, all possible outputs have the same probability
\end{description}

\vspace{-0.2in}

\begin{enumerate}[leftmargin=!,labelwidth=1.5em,itemsep=0.0in, label=\arabic*]

\item\label{step:generate_data} In the notation of Theorem~\ref{thm:k_canalizing}, generate the following:
  \begin{enumerate}[itemsep=0pt]
      \item random bits $b, a_1, \ldots, a_n \in \{0, 1\}$;
      \item a random subset $X \subset \{x_1, \ldots, x_n\}$ with $|X| = k$;
      \item a random ordered partition $X = X_1 \sqcup \ldots \sqcup X_r$ of $X$ (using Algorithm~\ref{alg:partition});
      \item a random noncanalizing function $p_C \not\equiv 0$ in variables $\{x_1, \ldots, x_n\} \setminus X$ (see Remark~\ref{rem:noncanalizing}).
  \end{enumerate}
  
  \item\label{alg_step:form_function} Form a function $f(x_1, \ldots, x_n)$ using the data generated in Step~\ref{step:generate_data} as in Theorem~\ref{thm:k_canalizing} where $M_i$ involves exactly the variables from $X_i$ for every $1 \leqslant i \leqslant r$.
  
  \item\label{alg_step:check} If $f$ does not satisfy any of the conditions~\ref{e1} or~\ref{e2}, discard it and run the algorithm again.
  Otherwise, return $f$.
\end{enumerate}
\end{algorithm}

\begin{remark}\label{rem:noncanalizing}
 We generate a random noncanalizing function as follows. 
 We generate a random Boolean function and test for canalization until we generate a noncanalizing one.
 Then we return it.
 Since canalizing functions are rare~\cite[Section~5]{HM16}, this algorithm is fast enough for our purposes (see Lemma~\ref{lem:complexity_generate_noncanalizing}).
\end{remark}
 
\begin{algorithm}[H]
\caption{Generating a random ordered partition of a given finite set}\label{alg:partition}
\begin{description}[leftmargin=2.5em, itemsep=0pt]
  \item[In:] A finite set $X$ with $|X| = k$
 \item[Out:] An ordered partition $X = X_1 \sqcup \cdots \sqcup X_r$ into nonempty subsets $X_1,\ldots, X_r$ such that, for a fixed $X$, all possible outputs have the same probability
\end{description}

\vspace{-0.2in}

\begin{enumerate}[leftmargin=!,labelwidth=1.5em,itemsep=0.0in, label=\arabic*]

    \item\label{alg_step:compute_pk} Compute $p_0, \ldots, p_k$, where $p_i$ is the number of ordered partitions of a set of size $i$, using the recurrence $p_j=\sum\limits_{i = 0}^{j - 1}\binom{j}{i}p_{j - i}$, $p_0=1$ (see~\cite[Eq. (9)]{fubini}).
    
    \item\label{alg_step:gen_rnd} Generate an integer $N$ uniformly at random from $[1, p_k]$.

    \item\label{alg_step:find_j} Find the minimum integer $j$ between $1$ and $k$ such that $\sum\limits_{i = 0}^{j - 1} \binom{k}{i}p_{k - i} \geq N$.
    
    \item\label{alg_step:select_subset} Randomly select a subset $X_1 \subset X$ of size $j$. 
    
    \item Generate an ordered partition $X_2 \sqcup \cdots \sqcup X_r$ of $X \setminus X_1$ recursively.
    
    \item Return $X_1 \sqcup \cdots \sqcup X_r$.

\end{enumerate}
\end{algorithm}


\section{Simulations: results}\label{subsec:results}

\begin{notation}
For a Boolean network $\mathbf{f} = (f_1, \ldots, f_n)$, let $N(\mathbf{f})$ and $S(\mathbf{f})$ denote the number of the attractors of $\mathbf{f}$ and the sum of the sizes of the attractors of $\mathbf{f}$, respectively.
We define the average size of an attractor as $\AS(\mathbf{f}) := \frac{S(\mathbf{f})}{N(\mathbf{f})}$.
\end{notation}


\subsection{Sample means of $N(\mathbf{f})$ and $\AS(\mathbf{f})$}\label{subsec:means}

For every $n = 4, \ldots, 20$ and every $0 \leqslant k \leqslant n$, we generate random Boolean networks in $n$ variables of canalizing depth $k$ and compute the mean of $N(\mathbf{f})$ and $\AS(\mathbf{f})$.
Figure~\ref{fig:means} shows how these means depend on $k$ for $n = 15$ (based on $50000$ samples for each $k$).
The shape of the plots is similar for other values of $n$ we did computation for (that is, $n = 4, \ldots, 20$).
Note that although both means are decreasing, the decrease of the mean of $\AS(\mathbf{f})$ is more substantial.

\begin{figure}[H]
  \centering
  \subfloat[The number of attractors ($N(\mathbf{f})$)]{\label{fig:num_att}\includegraphics[width=0.47\textwidth]{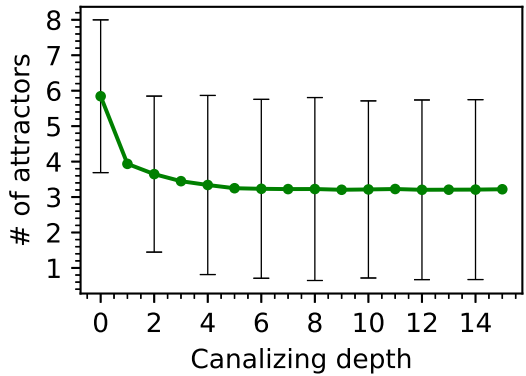}}
  \hspace{7mm}
  \subfloat[Average size of an attractor ($\AS(\mathbf{f})$)]{\label{fig:avg_num_att}\includegraphics[width=0.47\textwidth]{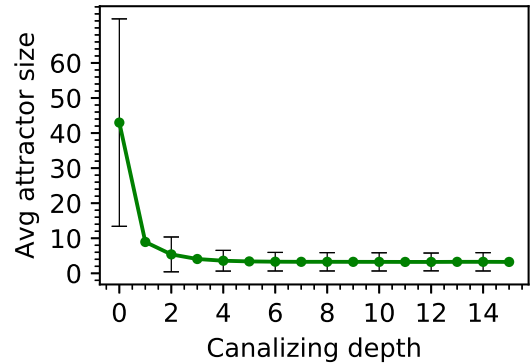}}
  
  \caption{Dependencies of the sample means of $N(\mathbf{f})$  and $\AS(\mathbf{f})$ on the canalizing depth}\label{fig:means}
\end{figure}


\subsection{Distributions of $N(\mathbf{f})$ and $\AS(\mathbf{f})$}\label{subsec:freq}

Figure~\ref{fig:freq} shows the empirical distributions of $N(\mathbf{f})$ and $\AS(\mathbf{f})$ for $n = 12$ and $k = 0, 1, 3, 12$ based on $300000$ samples for each $k$.
From the plot, we can make the following observations.
\begin{itemize}
    \item The distributions become more concentrated and the peak shifts towards zero when $k$ increases.
    \item 
    The distributions for nonzero canalizing depths (especially for larger depths) are much closer to each other that to the distribution for zero canalizing depth.
    This agrees with the plots on Figure~\ref{fig:means}.
\end{itemize}

\begin{figure}[H]
  \centering
  \subfloat[Distribution of the number of attractors ($N(\mathbf{f})$)]{\includegraphics[width=0.45\textwidth]{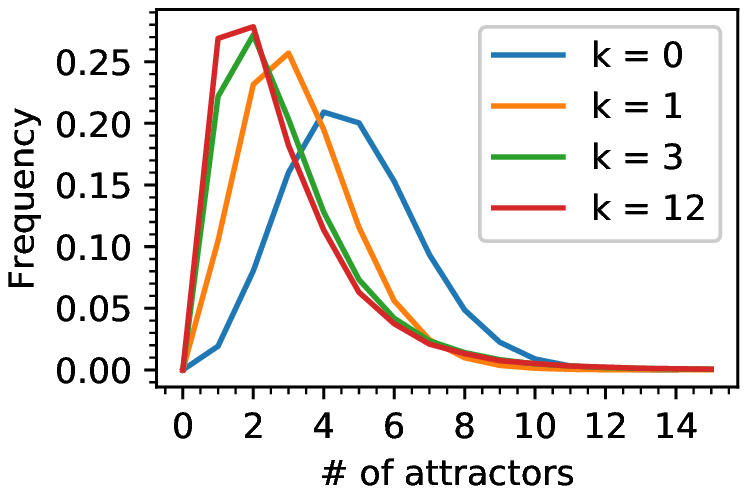}}
  \hspace{7mm}
  \subfloat[Distribution of the average size of an attractor~($\AS(\mathbf{f})$)]{\includegraphics[width=0.45\textwidth]{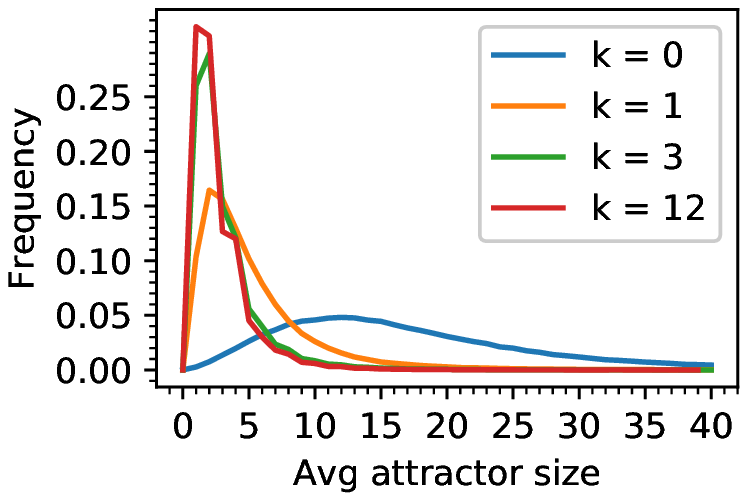}}
  
  \caption{Empirical distributions of $N(\mathbf{f})$  and $\AS(\mathbf{f})$ for $n = 12$ and $k = 0, 1, 3, 12$}\label{fig:freq}
\end{figure}


\subsection{Relative decreases}\label{subsec:ratio}

From Figure~\ref{fig:means}, we can observe that, for both $N(\mathbf{f})$ and $\AS(\mathbf{f})$, the sample mean decreases rapidly for small canalizing depths.
In order to understand how this decrease behaves for large $n$, we introduce
\[
N_{k}(n) := \frac{\text{the sample mean of $N(\mathbf{f})$ for $n$ variables and canalizing depth $k$}}{\text{the sample mean of $N(\mathbf{f})$ for $n$ variables and canalizing depth $0$}}.
\]
$\AS_{k}(n)$ is defined analogously.
Figure~\ref{fig:drop} plots $N_{1}(n)$, $N_{2}(n)$, $N_{3}(n)$, $N_n(n)$ and $\AS_1(n)$, $\AS_2(n)$, $\AS_3(n)$, $\AS_n(n)$ as functions of~$n$.
From the plots we see that
\begin{itemize}
    \item the relative initial decease from canalizing depth $0$ to canalizing depth $1$ becomes even more substantial when $n$ increases;
    \item the relative decrease from canalizing depth $0$ to canalizing depth $3$ is already very close to the relative decrease from depth zero to the maximal depth (i.e., nested canalizing functions).
\end{itemize}

\begin{figure}[H]
  \centering
  \subfloat[Relative decrease of the number of attractors]{\label{fig:num_att}\includegraphics[width=0.45\textwidth]{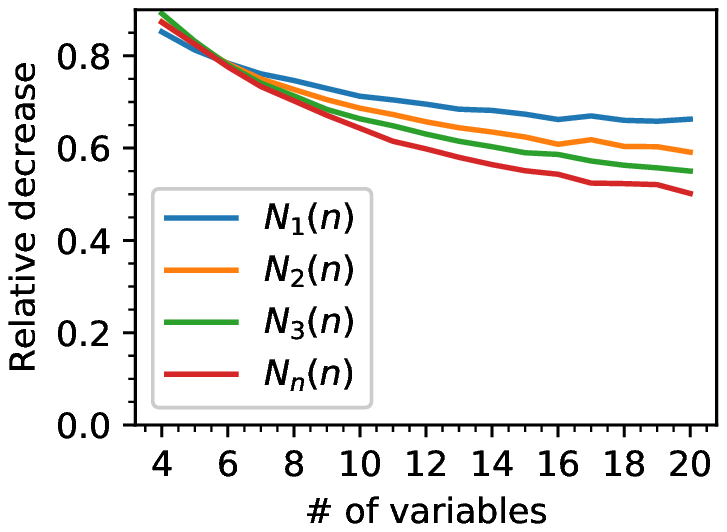}}
  \hspace{7mm}
  \subfloat[Relative decrease of the average size of an attractor]{\label{fig:avg_num_att}\includegraphics[width=0.45\textwidth]{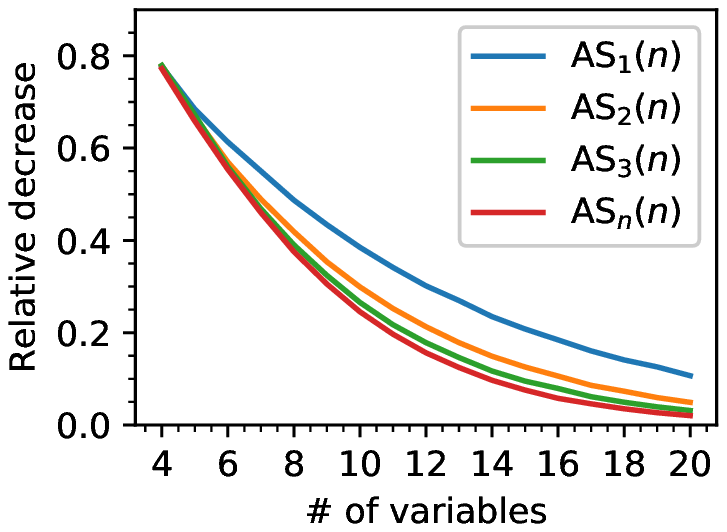}}
  
  \caption{Dependence of the relative decreases of the sample means of $N(\mathbf{f})$  and $\AS(\mathbf{f})$ on the number of variables $n$}\label{fig:drop}
\end{figure}


\section{Theory: the main result}\label{sec:th_result}

We will introduce notation needed to state the main theorem.
Let us fix a positive integer $\ell$.
For a binary string $\alpha \in S := \{0, 1\}^\ell$, we define:
\begin{itemize}
    \item $|\alpha|$ denotes the number of ones;
    \item $\bar{\alpha}$ denotes component-wise negation;
    \item $s(\alpha)$ denotes a cyclic shift to the right.
\end{itemize}
For binary strings $\alpha, \beta \in \{0, 1\}^\ell$ we define
\[
f(\alpha, \beta) := \begin{cases}
  \frac{1}{2^{|\beta|}}, \text{ if } \alpha \vee \beta = \beta,\\
  0, \text{otherwise}
\end{cases}
\quad\text{ and }\quad
g(\alpha, \beta) := \frac{1}{4}(f(\alpha, \beta) + f(\bar{\alpha}, \beta) + f(\alpha, \bar{\beta}) + f(\bar{\alpha}, \bar{\beta})).
\]
Then we define a $2^\ell \times 2^\ell$ matrix $G_\ell$ by
\begin{equation}\label{eq:G_ell}
  (G_\ell)_{a, b} = g(a, s(b)),    
\end{equation}
where we interpret numbers $1 \leqslant a, b \leqslant 2^\ell$ as binary sequences of length $\ell$.

\begin{theorem}\label{thm:main}
  Let $A_\ell$ be the limit of the expected number of attractors of length $\ell$ in a random Boolean network of canalizing depth one (see Definition~\ref{def:network_depth}) when the number of variables $n$ goes to infinity.
  Then
  \[
  A_\ell = \frac{1}{\ell P_{G_\ell}'(1)},
  \]
  where $P_{G_\ell}$ is the characteristic polynomial of matrix $G_\ell$ introduced above.
  In particular, we have
  \begin{align*}
      A_1 &= 1, &A_4 &= 0.2856\ldots\\
      A_2 &= \frac{2}{3} = 0.666\ldots, &A_5 &= 0.2004\ldots\\
      A_3 &= \frac{64}{189} = 0.3386\ldots, & A_6 &= 0.1721\ldots.\\
  \end{align*}
\end{theorem}

\begin{remark}
  The plots below show that the result of Theorem~\ref{thm:main} agrees with our simulations
  \begin{figure}[H]
  \centering
  \subfloat[Length 1]{\includegraphics[width=0.45\textwidth]{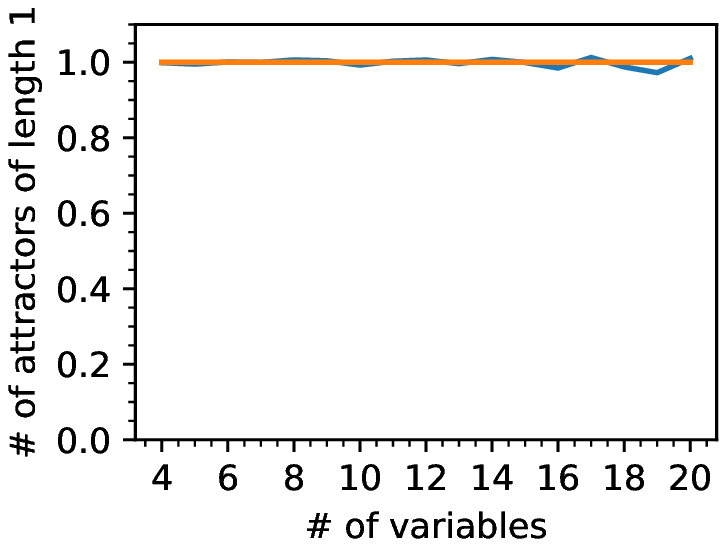}}
  \hspace{7mm}
  \subfloat[Length 2]{\includegraphics[width=0.45\textwidth]{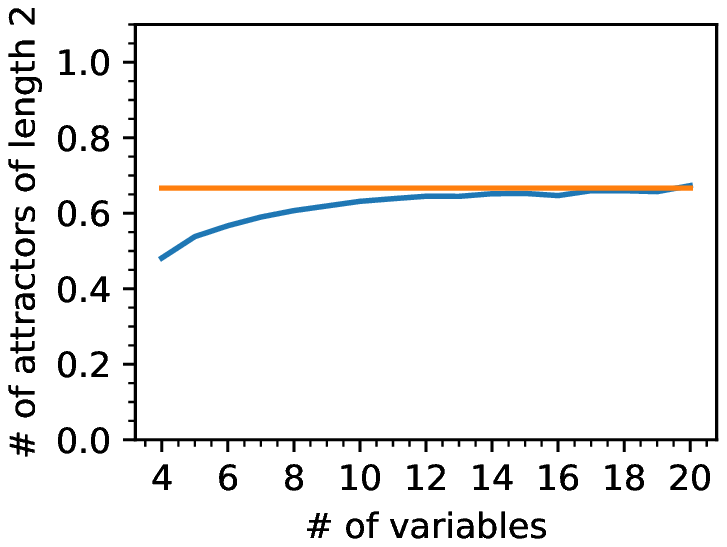}}
  
  \subfloat[Length 3]{\includegraphics[width=0.45\textwidth]{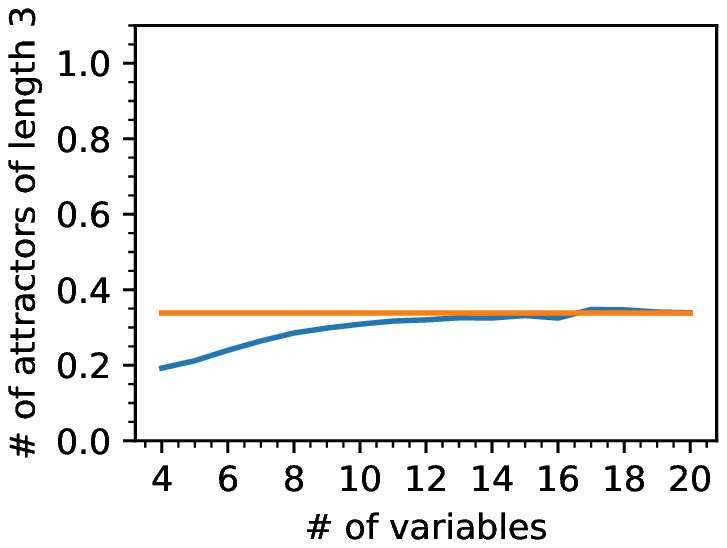}}
  \hspace{7mm}
  \subfloat[Length 4]{\includegraphics[width=0.45\textwidth]{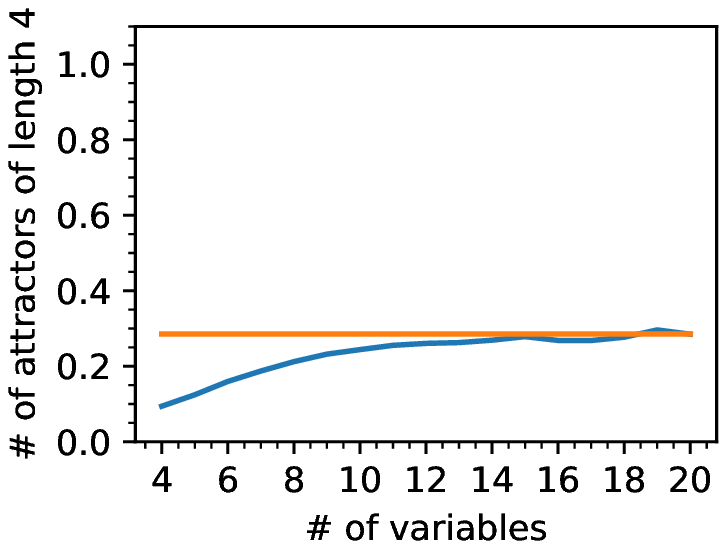}}

  \caption{The average number of attractors of fixed length (blue plot) compared to the limiting value from Theorem~\ref{thm:main} (orange plot)}
\end{figure}
\end{remark}

\begin{remark}\label{rem:at_least_one}
  As explained in Remark~\ref{rem:distr_C}, Theorem~\ref{thm:main} stills holds if we replace a random Boolean network of canalizing depth one with a random Boolean network defined by canalizing functions.
\end{remark}

\begin{example}
  Let $\ell = 2$.
  Then, for example, we have $f(0, 2) = f(0, 1) = \frac{1}{2}$ and $g(0, 1) = g(3, 1) = \frac{1}{4}$.
  In total, we have
  \[
    G_2 = \begin{pmatrix}
      3/8 & 1/4 & 1/4 & 3/8 \\
      1/8 & 1/4 & 1/4 & 1/8 \\
      1/8 & 1/4 & 1/4 & 1/8 \\
      3/8 & 1/4 & 1/4 & 3/8
    \end{pmatrix}
    \quad\text{ and } \quad P_{G_2}(t) = t^4 - \frac{5}{4}t^3 + \frac{1}{2}t^2.
  \]
\end{example}

\begin{remark}\label{rem:compare_with_depth_zero}
  Theorem~\ref{thm:main} and Corollary~\ref{cor:less_than_one} imply that $A_\ell > \frac{1}{\ell}$ for every $\ell > 1$.
  On the other hand, a direct computation shows that the expected number of attractors of length $\ell$ in a random Boolean network (without any canalization requirements) is $\frac{1}{\ell}$.
  This is consistent with our observations from Section~\ref{subsec:means}. 
\end{remark}

\begin{remark}
  A Sage script for computing numbers $A_\ell$ is available at \url{https://github.com/MathTauAthogen/Canalizing-Depth-Dynamics/blob/master/core/theory.sage}.
\end{remark}


\section{Conclusions}\label{sec:conclusions}
 
We conducted computational experiments to investigate the attractor structure of Boolean networks defined by functions of varying canalizing depth. 
We observed that networks with higher canalizing depth tend to have fewer attractors and the sizes of the attractors decrease dramatically when the canalizing depth increases moderately.
As a consequence, the basins tend to grow when the canalizing depth increases.
These properties are desirable in many biological applications of Boolean networks, so our results give new indications of the biological utility of Boolean networks defined by functions of positive canalizing depth.

We proved a theoretical result, Theorem~\ref{thm:main}, which complements the above observation as follows. 
The theorem implies that a large random Boolean network of canalizing depth one has on average more attractors of small size than a random Boolean network of the same size although it has less attractors in total. 
This also provides an explanation to the fact that the total size of attractors decreases faster than the number of attractors as the canalizing depth grows.

Furthermore, we observed that all the statistics we computed are almost the same in the case of the maximal possible canalizing depth (so-called nested canalizing Boolean networks) and in the case of moderate canalizing depth. This agrees with the results of Layne, Dimitrova, and Macauley~\cite{LDM12}.This observation elucidates an interesting and powerful feature of canalization: even a very moderate canalizing influence
in a Boolean network has a strong constraining influence on network dynamics. It would be of interest to explore the prevalence of this features in published Boolean network models. 
 
Finally, we provided evidence that the observed phenomena will occur for Boolean networks with larger numbers of state variables.

\subsection*{Acknowledgments} 
The authors are grateful to Claus Kadelka, Christian Krattenthaler, and Doron Zeilberger for helpful discussions and to the referees for valuable suggestions and comments.
 GP was partially supported by NSF grants CCF-1564132, CCF-1563942, DMS-1760448, DMS-1853482, and DMS-1853650 by
PSC-CUNY grants \#69827-0047 and \#60098-0048.
RL was partially supported by Grants NIH 1U01EB024501-01 and NSF CBET-1750183.
EP, GP, and WQ are grateful to the New York Math Circle where their collaboration started.

\subsection*{Data Availability Statement}

Python/sage code and the results of simulations used to support the findings of this study have been deposited at \url{https://github.com/MathTauAthogen/Canalizing-Depth-Dynamics}.

\bibliographystyle{abbrvnat}
\bibliography{bibdata}


\section*{Appendix A: Proofs}

     \setcounter{lemma}{0}
     \renewcommand{\thelemma}{A.\arabic{lemma}}
     \setcounter{notation}{0}
     \renewcommand{\thenotation}{A.\arabic{notation}}
     \setcounter{remark}{0}
     \renewcommand{\theremark}{A.\arabic{remark}}
     \setcounter{corollary}{0}
     \renewcommand{\thecorollary}{A.\arabic{corollary}}
     \setcounter{example}{0}
     \renewcommand{\theexample}{A.\arabic{example}}

\begin{notation}\label{not:G}
  We fix a positive integer $\ell$.
  \begin{itemize}
  \item 
  For every $1 \leqslant i < j \leqslant \ell$, we define a subset $S_{i, j} \subset S = \{0, 1\}^\ell$ by
  \[
    S_{i, j} := \{ (\alpha_1, \ldots, \alpha_\ell) \in S \mid \alpha_i = \alpha_j \}.
  \]
  
  \item For every $1 \leqslant i < j \leqslant \ell$, let $G_{\ell; i, j}$ be the submatrix of $G_\ell$ with rows and columns having indices from $S_{i, j}$.
  \end{itemize}
\end{notation}
    
\begin{lemma}\label{lem:stochastic}
      For every $\ell$, we have
      \begin{enumerate}
          \item $G_\ell^T$ is stochastic (see~\cite[\S~8.5]{matrices}), and $G_\ell$ has exactly one eigenvalue being equal to $1$.
          \item for every $1 \leqslant i < j \leqslant \ell$, there exists a $2^{\ell - 1} \times 2^{\ell - 1}$-matrix $C_{\ell; i, j}$ with nonnegative entries such that $\frac{2^{\ell + 2}}{2^{\ell + 2} - 1} (G_{\ell; i, j} + C_{\ell; i, j})^T$ is stochastic, and has exactly one of the eigenvalues being equal to $1$. 
      \end{enumerate}
\end{lemma}

\begin{proof}
  We will first show that $G_\ell^T$ is stochastic and irreducible (see~\cite[\S~3.11]{matrices}).

   By definition, showing that $G_\ell^T$ is stochastic is equivalent to proving that, for every $\beta \in S := \{0, 1\}^\ell$,
   \[
   \sum\limits_{\alpha \in S} g(\alpha, s(\beta)) = 1.
   \]
   Since shift just permutes binary strings, this sum is equal to $\sum\limits_{\beta \in S} g(\alpha, \beta)$.
   For a fixed $\beta$ and $k \leqslant |\beta|$, the number of $\alpha \in S$ such $\alpha \vee \beta = \beta$ and $|\alpha| = k$ is equal to $\binom{|\beta|}{k}$.
   Thus 
   \[
   \sum\limits_{\alpha \in S} h(\alpha, \beta) = \sum\limits_{k = 0}^{|\beta|} \binom{|\beta|}{k} \frac{1}{2^{|\beta|}} = 1 \quad\implies\quad \sum\limits_{\beta \in S} g(\alpha, \beta) = 1. 
   \]
   To prove irreducibility, we observe that, if $\bm{0} \in S$ denotes a zero binary string, then $g(\alpha, \bm{0}) \neq 0$ and $g(\bm{0}, \alpha) \neq 0$ for every $\alpha \in S$.
   Then~\cite[\S~3.11, Exercise~12a]{matrices} implies that $G_\ell^T$ is irreducible.
   
   Since $G_\ell^T$ is stochastic, its largest eigenvalue is equal to $1$~\cite[\S~8.5, p. 156]{matrices}.
   Since $G_\ell^T$ is irreducible, \cite[Theorem~8.2]{matrices} implies that $1$ is a simple eigenvalue.
   
   To prove the second part of the lemma, we fix $1 \leqslant i < j \leqslant \ell$.
   We will show that for every $\beta \in S_{i, j}$
   \begin{equation}\label{eq:col_sum}
     \sum\limits_{\alpha \in S_{i, j}} g(\alpha, s(\beta)) \leqslant \frac{2^{\ell + 2} - 1}{2^{\ell + 2}}.
   \end{equation}
   Indeed, let $\gamma$ be a binary string with all zeroes and one at the $i$-th position.
   Then, since $g(\gamma, s(\beta)) \geqslant \frac{1}{2^{|\beta| + 2}} \geqslant \frac{1}{2^{\ell + 2}}$, we have
   \[
     \sum\limits_{\alpha \in S_{i, j}} g(\alpha, s(\beta)) \leqslant \left( \sum\limits_{\alpha \in S} g(\alpha, s(\beta)) \right) - g(\gamma, s(\beta)) \leqslant 1 - \frac{1}{2^{\ell + 2}}.
   \]
   Inequality~\eqref{eq:col_sum} implies that there exists a matrix $C_{\ell; i, j}$ with nonnegative entries such that $\frac{2^{\ell + 2}}{2^{\ell + 2} - 1} (G_{\ell; i, j} + C_{\ell; i, j})^T$ is stochastic.
   
   Since $\bm{0} \in S_{i, j}$, the same argument as in the proof of the first part of the lemma shows that $\frac{2^{\ell + 2}}{2^{\ell + 2} - 1} (G_{\ell; i, j} + C_{\ell; i, j})^T$ is stochastic, and has exactly one of the eigenvalues being equal to $1$.
\end{proof}

\begin{corollary}\label{cor:less_than_one}
  Let $P_\ell(t)$ be the charactersitic polynomial of $G_\ell$. 
  Then, for every $\ell > 1$, $|P_\ell'(1)| < 1$.
\end{corollary}

\begin{notation}
  Fix a positive integer $n$.
  For vectors $\bm{a} = (a_1, \ldots, a_n) \in \mathbb{Z}_{\geqslant 0}^n$ and $\bm{b} = (b_1, \ldots, b_n) \in \mathbb{Z}_{\geqslant 0}^n$, we denote
   \[
   \bm{a}! := a_1! \cdot\ldots\cdot a_n!, \quad \bm{a}^{\bm{b}} := a_1^{b_1}\cdot \ldots\cdot a_n^{b_n}, \quad |\bm{a}| := a_1 + \ldots + a_n.
   \]
\end{notation}

\begin{lemma}\label{lem:lim}
   Let $A$ be an $s\times s$ stochastic matrix with only one of the eigenvalues being one.
   We set
   \begin{equation}\label{eq:BAn}
   C(A)_n := \sum\limits_{\substack{\mathbf{m} \in \mathbb{Z}_{\geqslant 0}^s \\ |\mathbf{m}| = n}} \frac{(A\mathbf{m})^{\mathbf{m}}}{n^n}.
   \end{equation}
   Let $P_A(t)$ be the characteristic polynomial of $A$.
   Then $\lim\limits_{n \to \infty} C(A)_n = \frac{1}{P_A'(1)}$.
\end{lemma}

\begin{proof}
      We recall that the Lambert $W$ function~\cite{LabertW} is the principal branch of the inverse of $x e^x$.
   We will use the notation $y(z) = -W(-z)$ from~\cite{RamQFunc} so that $y(z) = ze^{y(z)}$.
   Function $y(z)$ has a singularity of the square-root type at $z = 1/e$ and has the following expansion around this point (see~\cite[p.~107]{RamQFunc})
   \begin{equation}\label{eq:y_around_sing}
   y(z) = 1 - \varepsilon + \frac{1}{3}\varepsilon^2 - \ldots, \quad \text{ where } \varepsilon = \sqrt{2 - 2ez}.
   \end{equation}
   From this, we obtain
   \begin{equation}\label{eq:y_around_sing2}
   \frac{1}{y(z)} = 1 + \varepsilon + \frac{2}{3}\varepsilon^2 - \ldots, \quad \text{ where } \varepsilon = \sqrt{2 - 2ez}.
   \end{equation}
   The main result of~\cite{Carlitz} implies that, for every complex $s\times s$ matrix $A$, we have
   \begin{equation}\label{eq:Carlitz}
   \sum\limits_{\mathbf{m} \in \mathbb{Z}_{\geqslant 0}^s} \frac{(A\mathbf{m})^{\mathbf{m}}}{\mathbf{m}!} x^{|\mathbf{m}|} \exp\bigl( -x\sum\limits_{i, j}m_ja_{i, j} \bigr) = \frac{1}{\det|E - xA|}.
   \end{equation}
   Since $A^T$ is stochastic, we have $\sum\limits_{i = 1}^n a_{i, j} = 1$, so
   \begin{equation}\label{eq:Carlitz_stoch}
   \sum\limits_{\mathbf{m} \in \mathbb{Z}_{\geqslant 0}^s} \frac{(A\mathbf{m})^{\mathbf{m}}}{\mathbf{m}!} x^{|\mathbf{m}|} e^{-x|\mathbf{m}|} = \frac{1}{\det|E - xA|}.
   \end{equation}
   If we perform a substitution $x = y(z)$ and use the definition of the Lambert $W$ function, we obtain
   \begin{equation}\label{eq:Carlitz_final}
   \sum\limits_{\mathbf{m} \in \mathbb{Z}_{\geqslant 0}^s} \frac{(A\mathbf{m})^{\mathbf{m}}}{\mathbf{m}!} z^{|\mathbf{m}|} = \frac{1}{\det|E - y(z)A|}.
   \end{equation}
   From this, we obtain
   \begin{equation}\label{eq:gen_func}
   \sum\limits_{n = 0}^\infty \frac{n^n C(A)_n}{n!} z^n = \sum\limits_{\mathbf{m} \in \mathbb{Z}_{\geqslant 0}^s} \frac{(A\mathbf{m})^{\mathbf{m}}}{\mathbf{m}!} z^{|\mathbf{m}|} = \frac{1}{\det|E - y(z)A|} =: F(z). 
   \end{equation}
   $F(z)$ can be rewritten as
   \[
     F(z) = \frac{1}{y(z)^{s} P_A(1 / y(z))}.
   \]
   Finding the asymptotic behavior of the Taylor coefficients of $F(z)$ would yield an asymptotic for $C(A)_n$.
   We will do this using singularity analysis~\cite[Chapter~VI]{FS} (similarly to~\cite[Theorem~2]{RamQFunc}).
   Since $|y(z)| < 1$ for $|z| < 1 / e$ (see~\cite[Fig.~1]{LabertW}) and all roots of $P_A$ lie in the unit circle due to the stochasticity of $A$, $\frac{1}{e}$ is the singularity of $F(z)$ with the smallest absolute value.
   Due to Lemma~\ref{lem:stochastic}, $P_A(t) = (1 - t) Q_A(t)$, where $Q_A(1) \neq 0$.
   Using~\eqref{eq:y_around_sing}, we obtain the following expansion of $F(z)$ around $1 / e$:
   \[
     F(z) = \frac{1}{(1 - \varepsilon + \ldots)^{s} (-\varepsilon - \frac{2}{3} \varepsilon^2 + \ldots) Q_A(1 + \varepsilon + \ldots)} = \frac{-1}{Q_A(1)} (1 / \varepsilon + \ldots), \text{ where } \varepsilon = \sqrt{2 - 2ez}.
   \]
   Singularity analysis~\cite[Corollary~VI.1]{FS} implies that
   \[
   \frac{n^nC(A)_{n}}{n!} \sim \frac{-e^n}{Q_A(1) \sqrt{2\pi n}} \quad\text{ as } n \to \infty.
   \]
   Using Stirling's formula, we get 
   \[
   C(A)_{n} \sim \frac{-n! e^n}{n^n Q_A(1) \sqrt{2\pi n}} \sim \frac{-1}{Q_{A}(1)} \quad\text{ as } n\to \infty.
   \]
   Using $P_A' = -Q_A' + (1 - t)Q_\ell$, we deduce $P_A'(1) = -Q_{A}(1)$, and this finishes the proof.
\end{proof}


\begin{lemma}\label{lem:eq_lim}
  On the set of all Boolean networks with $n$ states consider two probability distributions:
  \begin{enumerate}[label = (\Alph*)]
      \item\label{probab:our} all the networks with canalizing depth one have the same probability, all others have probability zero;
      \item\label{probab:convenient} the probability assigned to each network is proportional to the product of the number of canalizing variables of the functions defining this network.
  \end{enumerate}
  We fix a positive integer $\ell$.
  By $A_{\ell, n}$ and $B_{\ell, n}$ we denote the average number of attractors of length $\ell$ in a random Boolean network with $n$ states with respect to distributions~\ref{probab:our} and~\ref{probab:convenient}, respectively.
  Then
  \[
  \lim\limits_{n \to \infty} A_{\ell, n} = \lim\limits_{n \to \infty} B_{\ell, n}.
  \]
\end{lemma}

\begin{example}
  We will illustrate the distribution~\ref{probab:convenient} by an example. 
  Consider three following networks with two states:
  \[
    \mathbf{f}_1 = (x_1x_2 + 1, x_1 + x_2), \quad \mathbf{f}_2 = (x_1x_2, x_1), \quad \text{ and } \quad \mathbf{f}_3 = (x_1x_2 + 1, x_1x_2).
  \]
  Since the canalizing depth of $x_1 + x_2$ is zero, $P_B(\mathbf{f}_1)$, the probability of $\mathbf{f}_1$ with respect to $B$, is zero.
  Since the canalizing depths of $x_1x_2$ and $x_1$ are $2$ and $1$, respectively, the ratio $P_B(\mathbf{f}_2) / P_B(\mathbf{f}_3)$ is equal to $\frac{2\cdot 1}{2\cdot 2} = 1/2$.
\end{example}

\begin{proof}
  Let $F_n$ and $F^\ast_n$ be the number of Boolean functions in $n$ variables with of canalizing depth exactly one and more than one, respectively.
  We will use the following bounds:
  \begin{enumerate}
      \item $F^\ast_n \leqslant n^2 \cdot 4 \cdot 4 \cdot 2^{2^{n-2}}$. 
      We look term-by-term. 
      There are at most $n^2$ ways to choose first and second canalizing variables.
      There are at most $4$ choices for the canalizing outputs and at most $4$ choices for canalizing values for these two variables.
      There are at most $2^{2^{n-2}}$ core functions, since that is all possible functions, which may or may not be canalizing. 
      Since redundant arrangements of canalizing variables are not accounted for, this must overcount.
      \item $F_n \geqslant 2^{2^{n-1}} - (n-1) \cdot 2 \cdot 2 \cdot 2^{2^{n-2}}$.
      This is a lower bound for the number of non-canalizing core function in $n - 1$ variables because $(n-1) \cdot 2 \cdot 2 \cdot 2^{2^{n-2}}$ is an upper bound on the number of canalizing functions in $n - 1$ variables (obtained in the same way as the bound above).
  \end{enumerate}
  We also introduce
  \begin{equation}\label{eq:Rn}
  R_n := \frac{F^\ast_n}{F_n} \leqslant \dfrac{16  n^2  2^{2^{n - 2}}}{2^{2^{n - 1}} - 4(n - 1) 2^{2^{n - 2}}} = \dfrac{n^2}{2^{(2^{n-2}) - 4}-\frac{1}{4}(n-1)}.
  \end{equation}

For $X$ being~\ref{probab:our} or~\ref{probab:convenient} and positive integer $n$, let $P_{X, n}$ denote the probability (it is always the same) of choosing a network from distribution $X$ with all functions being of depth exactly one.
Let $P_n^\ast$ be the maximal probability of choosing a network from~\ref{probab:convenient} with at least one function being of depth more than one, respectively.
By $S_n$ and $S_n^\ast$ we denote the total number of attractors of length $\ell$ in networks with all functions being of depth exactly one and with at least one function being of depth more than one, respectively.

The statement of the lemma is equivalent to the statement that 
\begin{equation}\label{eq:main_lim}
\lim\limits_{n\to\infty} (A_{\ell, n} - B_{\ell, n}) = 0
\end{equation}
Using the notation introduced above, we can bound $A_{\ell, n} - B_{\ell, n}$ as
\begin{equation}\label{eq:bound}
    P_{n, A}S_n - P_{n, B} S_n - P_n^\ast S_n^\ast \leqslant |A_{\ell, n} - B_{\ell, n}| \leqslant P_{n, A} S_n + P_{n, B} S_n
\end{equation}
We set $U_n := S_n(P_{n, A} - P_{n, B})$ and $V_n := P_n^\ast S_n^\ast$.
Then~\eqref{eq:main_lim} would follow from $\lim\limits_{n \to \infty} U_n = 0$ and $\lim\limits_{n \to \infty} V_n = 0$, so we will prove these two equalities.

Since any network has at most $2^n$ attractors of length $\ell$, $S_n \leqslant 2^nF_n^n$.
Since the total sum of the products of canalizing depths over all Boolean networks does not exceed $(F_n + nF_n^\ast)^n$, we have $P_{n, B} \geqslant \frac{1}{(F_n + nF_n^\ast)^n}$. 
Since $P_{n, A} = \frac{1}{F_n^n}$, we have
\[
U_n \leqslant 2^nF_n^n \left(\frac{1}{F_n^n}-\frac{1}{(F_n + nF_n^\ast)^n}\right) = 2^n \left(1 - \frac{1}{(1 + nR_n)^n} \right) = 2^n \frac{\binom{n}{1}nR_n + \binom{n}{2}(nR_n)^2 + \ldots + (nR_n)^n}{(1 + nR_n)^n}.
\]
\eqref{eq:Rn} implies that $nR_n < 1$ for large enough $n$.
Hence, for large enough $n$, we have
\[
U_n \leqslant 2^n nR_n \frac{2^n}{(1 + nR_n)^n} \leqslant 4^nnR_n \leqslant \dfrac{4^n n^3}{2^{(2^{n-2}) - 4}-\frac{1}{4}(n-1)} \to 0.
\]
By similar arguments, $P_n^*\leqslant \frac{n^n}{F_n^n}$ and $S_n^* \leqslant 2^n n(F_n + F_n^\ast)^{n - 1}F_n^\ast$ so:
\[
V_n \leqslant 2^n n^{n + 1} (F_n + F_n^\ast)^{n - 1}F_n^\ast \frac{1}{F_n^n} \leqslant 2^n n^{n + 1} (1 + R_n)^{n - 1} R_n.
\]
Since $R_n < 1$ for large enough $n$, using~\eqref{eq:Rn}, we have
\[
V_n \leqslant 2^{2n - 1} n^{n + 1} R_n \leqslant \dfrac{2^{2n - 1} n^{n + 3}}{2^{(2^{n-2}) - 4}-\frac{1}{4}(n-1)} \to 0.\qedhere
\]
\end{proof}

\begin{remark}\label{rem:distr_C}
  The proof of Lemma~\ref{lem:eq_lim} will be valid if we replace distribution~\ref{probab:convenient} with any other distribution $(C)$ such that, for every Boolean network $\mathbf{f} = (f_1, \ldots, f_n)$, 
  \begin{itemize}
      \item if at least one of $f_i$'s is non-canalizing, $P_C(\mathbf{f}) = 0$;
      \item there exists a constant $P_{n, C}$ such that, if the canalizing depth of every $f_i$ is one, then $P_C(\mathbf{f}) = P_{n, C}$;
      \item we have $\frac{P_C(\mathbf{f})}{P_{n, C}} \leqslant \frac{P_{B}(\mathbf{f})}{P_{n, B}}$ (using notation from the proof of Lemma~\ref{lem:eq_lim}).
  \end{itemize}
  The above properties hold, for example, for the following distribution
  \begin{enumerate}[label = (\Alph*)]
      \item[(C)] all the networks defined by canalizing functions have the same probability, all others have probability zero.
  \end{enumerate}
  Using this distribution instead of~\ref{probab:our}, we see that Theorem~\ref{thm:main} holds also for a random Boolean network defined by canalizing functions.
\end{remark}


\begin{lemma}\label{lem:represent_B}
   We will use Notation~\ref{not:G} and notation from Lemma~\ref{lem:lim}.
   Then, for every positive integers $\ell$ and $n$, we have
   \begin{equation}\label{eq:sandwich}
     C(G_\ell)_n - \sum\limits_{1\leqslant i < j \leqslant \ell} C(G_{\ell; i, j})_n \leqslant \ell B_{\ell, n} \leqslant C(G_\ell)_n.
   \end{equation}
\end{lemma}

\begin{proof}
    We fix $n$.
  Consider a tuple $\bm{X} = (X_1, \ldots, X_\ell)$ of $\ell$ distinct elements of $\{0, 1\}^n$.
  For $1 \leqslant i \leqslant n$, we denote $\bm{X}_i := (X_{1, i}, \ldots, X_{n, i})$.
  For $\alpha \in S$, let
  \[
    n_\alpha := \left\lvert \{i \mid 1 \leqslant i \leqslant n,\; \bm{X}_i = \alpha\} \right\rvert.
  \]
  Then $\sum\limits_{\alpha \in S} n_\alpha = n$. 
  First, we will show that  
  \begin{equation}\label{eq:attractor_probab}
  P(X_1, \ldots, X_\ell \text{ form an attractor in this order}) =  \prod\limits_{\alpha \in S} \left(  \sum\limits_{\beta \in S} g(\alpha, s(\beta)) \frac{n_{\beta}}{n} \right)^{n_{\alpha}} = \frac{(G_\ell \mathbf{n})^{\mathbf{n}}}{n^n},
  \end{equation}
  where $\mathbf{n} = (n_0, n_1, \ldots, n_{2^\ell - 1})$.
  
  To prove~\eqref{eq:attractor_probab}, we will use that the functions $f_i$ ($i = 1, \ldots, n$) in the network are chosen independently to decompose the left-hand side as
  \[
  P(X_1, \ldots, X_\ell \text{ form an attractor in this order}) = \prod\limits_{i = 1}^n P(f_i(X_{j}) = X_{j + 1, i} \text{ for every } 1 \leqslant j \leqslant n),
  \]
  where we use notation $X_{n + 1} = X_1$ and the probability of each Boolean function to be chosen is assumed to be proportional to the number of its canalizing variables.
  We show that, for every $1 \leqslant i \leqslant n$,
  \begin{equation}\label{eq:single_function}
    P(f_i(X_{j}) = X_{j + 1, i} \text{ for every } 1 \leqslant j \leqslant n) = \sum\limits_{\beta \in S} g(\bm{X}_i, s(\beta)) \frac{n_{\beta}}{n}.
  \end{equation}
  Then~\eqref{eq:attractor_probab} would follow from multiplying~\eqref{eq:single_function} for all $i$.
  To prove~\eqref{eq:single_function}, without loss of generality, we consider $i = 1$.
  Consider a set 
  \[
  \Omega = \{ (f, k) \mid f\colon \{0, 1\}^n \to \{0, 1\},\; 1\leqslant k \leqslant n,\; x_k \text{ is canalizing for }f \}
  \]
  with a uniform probability distribution $P_{\Omega}$.
  Observe that for a function $f$ with canalizing variables $x_{k_1}$, $\ldots$, $x_{k_s}$, we have
  \[
  P(f) = P_\Omega((f, k_1)) + \ldots + P_\Omega((f, k_s)).
  \]
  
  If we can show that, for every $1 \leqslant k \leqslant n$,  
  \begin{equation}\label{eq:single_beta}
    P_\Omega(f(X_{j}) = X_{j + 1, 1} \text{ for every } 1 \leqslant j \leqslant n \mid (f, k) \in \Omega) = g(\bm{X}_1, s(\bm{X}_k)),
  \end{equation}
  then~\eqref{eq:single_function} would follow by summing up~\eqref{eq:single_beta} over all $k$ and using the law of total probability.
  
  We consider one of the canalizing variables of $f$, say, $x_k$.
  Let $c$ be the canalizing value of $x_{k_1}$, and let $v$ be the value taken by $f$ when $x_{k_1} = c$.
  Then $(c, v) \in \{0, 1\}^2$, and all these four cases have the same probability due to the symmetry.
   As $g(\alpha, s(\beta)) =\frac{1}{4}(h(\alpha, \beta) + h(\bar{\alpha}, \beta) + h(\alpha, \bar{\beta}) + h(\bar{\alpha}, \bar{\beta}))$, it is sufficient to show that 
   \begin{equation}\label{eq:cv_fixed}
      P_\Omega(f(X_{j}) = X_{j + 1, 1} \text{ for every } 1 \leqslant j \leqslant n \mid (f, k) \in \Omega \text{ and } c = v = 0) = h(\bm{X}_1, s(\bm{X}_k))
   \end{equation}
   and then sum for all $(c, v) \in \{0, 1\}^2$.
   
   To prove~\eqref{eq:cv_fixed}, consider any $j$, say $j = 1$.
   There are then 4 cases for the values of $X_{1, k}$ and $X_{2, 1}$.
   
  \begin{enumerate}[itemsep=0pt, topsep=0pt]
      \item \label{itm:genprf} $X_{1, k} = 1$ and $X_{2, 1}$ is $0$ or $1$. 
      With probability $\frac{1}{2}$, we have $f(X_1) = X_{2, 1}$. 
      This is true due to symmetry, as for any $f_1$ which takes on the value $w$ at $X_1$, we can produce another function $g$ that is equal to $0$ if $X_{1, k} = 0$ and $\bar{f}_1$ if $X_{1, k} = 1$. 
      Then $g(X_1) = \bar{w}$.
      
      \item\label{case:zero} $X_{1, k} = 0$ and $X_{2, 1} = 1$.
      Since $X_{1, k} = c$, the probability of $f(X_1) = X_{2, 1} \neq v = 0$ is zero.
      
      \item\label{case:one} $X_{1, k} = X_{2, 1} = 0$.
      Since $X_{1, k} = c$ and $X_{2, 1} = v$, the canalization property implies that $f(X_1) = X_{2, 1}$ with probability one.
  \end{enumerate}
  
  The only case in which $\bm{X}_1\vee s(\bm{X}_k)\neq s(\bm{X}_k)$ is where there is at least one $j$ such that case~\ref{case:zero} is realized.
  In this case, the probability in the left-hand side of~\eqref{eq:cv_fixed} will be zero.
  Otherwise, each occurrence of case~\ref{itm:genprf} will multiply the total probability by $\frac{1}{2}$  and each occurrence of case~\ref{case:one} will multiply the total probability by $1$. 
  Thus, we show that the left-hand side of~\eqref{eq:cv_fixed} is indeed equal to $h(\bm{X}_1, s(\bm{X}_k))$.
  This finishes the proof of~\eqref{eq:attractor_probab}.
  
  To finish the proof of the lemma, we set 
  \[
  U := \{\mathbf{n} \in \mathbb{Z}_{\geqslant 0}^{S} \mid \sum\limits_{\alpha \in S} n_{\alpha} = n \;\&\; \text{the support of } \mathbf{n} \text{ does not belong to } \bigcup\limits_{1 \leqslant i < j \leqslant \ell} S_{i, j}\}.
  \]
  Summing~\eqref{eq:attractor_probab} over all $\ell$-tuples $(X_1, \ldots, X_\ell)$ of distinct elements of $\{0, 1\}^n$, we obtain (see~\eqref{eq:BAn})
  \[
    \ell B_{\ell, n} = \sum\limits_{\mathbf{n} \in U} \frac{(G_\ell\mathbf{n})^{\mathbf{n}}}{n^n} \leqslant C(G_\ell)_n.
  \]
  On the other hand, if $\mathbf{n}$ is supported on one some $S_{i, j}$, then $G_\ell \mathbf{n} = G_{\ell;i, j} \mathbf{n}|_{S_{i, j}}$, where $\mathbf{n}|_{S_{i, j}}$ denotes the restriction of $\mathbf{n}$ on the coordinates from $S_{i, j}$.
  This implies that
  \[
  C(G_\ell)_n - \ell B_{\ell, n} \leqslant \sum\limits_{1 \leqslant i < j \leqslant \ell} C(G_{\ell; i, j})_n.
  \]
  This finishes the proof of the lemma.
\end{proof}


\begin{proof}[Proof of Theorem~\ref{thm:main}]
  We fix positive integer $\ell$. 
  In the notation of Lemma~\ref{lem:eq_lim}, we have $A_\ell = \lim\limits_{n \to \infty} A_{\ell, n}$.
  Lemma~\ref{lem:eq_lim} implies that $A_\ell = \lim\limits_{n \to \infty} B_{\ell, n}$.
  We fix any $1 \leqslant i < j \leqslant \ell$, and let $C_{\ell; i, j}$ be the matrix given by Lemma~\ref{lem:stochastic}.
  We set $M := \frac{2^{\ell + 2}}{2^{\ell + 2} - 1} (G_{\ell; i, j} + C_{\ell; i, j})$.
  Then
  \[
  0 \leqslant C(G_{\ell; i, j})_n \leqslant C(G_{\ell; i, j} + C_{\ell; i, j})_n = \left( \frac{2^{\ell + 2} - 1}{2^{\ell + 2}} \right)^n C(M)_{n}.
  \]
  Lemma~\ref{lem:lim} implies that $\lim\limits_{n \to \infty} C(M)_n$ is finite, thus we have that $\lim\limits_{n \to \infty} C(G_{\ell; i, j})_n = 0$.
  We finish the proof of the theorem by considering the limit of~\eqref{eq:sandwich} and applying Lemma~\ref{lem:lim} to $G_\ell$.
\end{proof}


\section*{Appendix B: Complexity analysis}

     \setcounter{lemma}{0}
     \renewcommand{\thelemma}{B.\arabic{lemma}}
     \setcounter{proposition}{0}
     \renewcommand{\theproposition}{B.\arabic{proposition}}
     \setcounter{notation}{0}
     \renewcommand{\thenotation}{B.\arabic{notation}}
     \setcounter{remark}{0}
     \renewcommand{\theremark}{B.\arabic{remark}}
     \setcounter{corollary}{0}
     \renewcommand{\thecorollary}{B.\arabic{corollary}}
     \setcounter{example}{0}
     \renewcommand{\theexample}{B.\arabic{example}}

   \begin{proposition}\label{prop:alg_partition}
     Complexity of Algorithm~\ref{alg:partition} is $\mathrm{O}(k^3)$.
   \end{proposition}

   \begin{proof}
             First we show the complexity of a single run the algorithm, i.e., not taking into account the recursive call, is $\mathrm{O}(k^2)$.

         First we show the complexity of a single run the algorithm, i.e., not taking into account the recursive call, is $\mathrm{O}(k^2)$.

         First we show the complexity of a single run the algorithm, i.e., not taking into account the recursive call, is $\mathrm{O}(k^2)$.

         First we show the complexity of a single run the algorithm, i.e., not taking into account the recursive call, is $\mathrm{O}(k^2)$.

     First we show the complexity of a single run the algorithm, i.e., not taking into account the recursive call, is $\mathrm{O}(k^2)$.
     Since the first $k$ rows of the Pascal's triangle can be precomputed in $\mathrm{O}(k^2)$, the complexity of step~\ref{alg_step:compute_pk} is also $\mathrm{O}(k^2)$.
     Similarly, the complexity of step~\ref{alg_step:find_j} is $\mathrm{O}(k^2)$.
     It remains to observe that step~\ref{alg_step:gen_rnd} takes $\mathrm{O}(1)$ and step~\ref{alg_step:select_subset} takes $\mathrm{O}(k^2)$ (indeed, selecting a subset of size $j$ amounts to selecting and removing $j$ indices).
     In total, we obtain $\mathrm{O}(k^2)$.
     
     The depth of the recursion calls is at most $k$. 
     Since the complexity of each single call is $O(k^2)$, so the total complexity is $O(k^3)$.
   \end{proof}

   \begin{lemma}\label{lem:complexity_generate_noncanalizing}
      The average complexity of the algorithm for generating a function in $n > 0$ variables which is either $1$ or noncanalizing described in Remark~\ref{rem:noncanalizing} is $\mathrm{O}(n2^n)$.
   \end{lemma}
   
   \begin{proof}
      \cite[p. 116]{canalizing_proportion} implies that the proportion of functions which are canalizing in $n$ variables is bounded from above by $\frac{4n}{2^{2^{n - 1}}}$.
      Note that \cite{canalizing_proportion} considers constant functions canalizing which we do not.
      Thus the probability $P_{n}$ of choosing a function which is either $1$ or noncanalizing is bounded from above by 
      \[
      \dfrac{4n}{2^{2^{n - 1}}} - \dfrac{1}{2^{2^{n}}} = \dfrac{4n - \frac{1}{2^{2^{n - 1}}}}{2^{2^{n - 1}}}.
      \]
      This bound is less than $\frac{3}{4}$ for all values of $n$ except 1 and 2, but we can compute directly that $P_1 = \frac{3}{4}$ and $P_2 = \frac{13}{16}$.
      Therefore, the number of times the generation of a function needs to be repeated averages to $\frac{1}{1 - P_{n}}$, which does not exceed $4$, so the average complexity of the whole procedure is the same as of a single generation step.
        
      The complexity of a single step consists of generating a random function (which is $\mathrm{O}(2^n)$) and checking whether it is canalizing or not.
      We perform this check by running linearly through the table for each variable, so the complexity is $\mathrm{O}(n2^n)$ time. 
      Thus, the total complexity is indeed $\mathrm{O}(n2^n)$.
   \end{proof}

   \begin{lemma}\label{lem:rerun}
      There is a constant $c < 1$ such that the probability that a function generated in steps~\ref{step:generate_data} and~\ref{alg_step:form_function} of Algorithm~\ref{alg:k_canalizing} does not satisfy one of the of the conditions~\ref{e1} or~\ref{e2} is bounded by $c$ for every $n$.
   \end{lemma}
   
   \begin{proof}
     Notice that
     \[
     P(\text{\ref{e1} or \ref{e2} is false}) = P(r \neq 1)P(\text{\ref{e1} is false} \mid r \neq 1) + P(r = 1)P(\text{\ref{e2} is false} \mid r = 1)
     \]
     We will show that there is a constant $c < 1$ such that $P(\text{\ref{e1} is false} \mid r \neq 1)$ and $P(\text{\ref{e2} is false} \mid r = 1)$ do not exceed $c$.
     
     \begin{itemize}
         \item $P(\text{\ref{e1} is false} \mid r \neq 1)$.
         The probability of having $k_r = 1$ (the only possible $k_r < 2$) is just the proportion of ordered partitions with a single element at the end. 
         We can construct all of these by picking an element and then picking a partition of the remaining elements, so this creates $k\cdot p_{k-1}$ possibilities. 
         Thus the probability this occurring is $\dfrac{kp_{k-1}}{p_k}$. 
         \cite[Eq. (5)]{fubini} implies that this approaches $\ln(2) < 1$ as $n$ goes to infinity. 
         Thus, there exists such $c$.
         
         \item $P(\text{\ref{e2} is false} \mid r = 1)$.
         The probability of ever picking $b=1$ is just $\frac{1}{2}$, so we can take $c = \frac{1}{2}$.
     \end{itemize}
   \end{proof}

   \begin{proposition}\label{prop:alg_k_can}
      Complexity of Algorithm~\ref{alg:k_canalizing} is $\mathrm{O}(n2^n)$.
   \end{proposition}

   \begin{proof}
     Lemma~\ref{lem:rerun} implies that the average number of reruns in step~\ref{alg_step:check} is constant.
     Thus, the complexity of the algorithm is the same as of a single run.
   
     Proposition~\ref{prop:alg_partition} and Lemma~\ref{lem:complexity_generate_noncanalizing} imply that the complexity of step~\ref{step:generate_data} is $\mathrm{O}(k^3 + (n - k)2^{n - k})$.
     Step~\ref{alg_step:form_function} generates a truth table for the function. 
     There are $2^n$ input-output pairs, and computing the function takes at most $k$ steps, so this is $\mathrm{O}(k2^n)$.
     In step~\ref{alg_step:check}, the conditions~\ref{e1} or~\ref{e2} are verified in $\mathrm{O}(2^n)$ time.
     
     Summing everything, we obtain $\mathrm{O}(k^3 + (n - k)2^{n - k} + k2^n) = \mathrm{O}(n2^n)$
   \end{proof}

\end{document}